\documentclass[preprint,12pt]{elsarticle}

\usepackage{amsmath}
\usepackage{amssymb}
\usepackage{mathscinet}
\usepackage{amsfonts}
\usepackage{epsfig}
\usepackage{amsthm}
\usepackage{latexsym}
\usepackage[latin1]{inputenc}
\usepackage{verbatim}
\usepackage{pb-diagram}
\usepackage[small]{caption} 
\usepackage[labelformat=empty]{subfig}
\usepackage{everysel, keyval, ragged2e}
\usepackage{wrapfig}

\newtheorem{theorem}{Theorem}
\newtheorem{proposition}{Proposition}
\newtheorem{conjecture}{Conjecture}
\newdefinition{definition}{Definition}
\newcommand{\m}{\mathcal }

\newcommand{\bra}[1]{\langle #1 |}    
\newcommand{\ket}[1]{| #1 \rangle}    
\def\R{\mathbb R}    

\def\N{\mathbb N}    
\newcommand\tr{ \mbox{\rm Tr} }

\journal{Linear Algebra and its Applications}

\begin{document}
\begin{frontmatter}

\title{Extending a characterization of majorization to infinite dimensions}

\author[uog]{Rajesh Pereira}
\address[uog]{Department of Mathematics \& Statistics, University of Guelph, Guelph, ON, Canada N1G 2W1}
\ead{pereirar@uoguelph.ca}

\author[bu]{Sarah Plosker\corref{cor1}}
\address[bu]{Department of Mathematics \& Computer Science, Brandon University, Brandon, MB, Canada R7A 6A9}
\ead{ploskers@brandonu.ca}

\cortext[cor1]{Corresponding author}

\begin{abstract}
We consider recent work linking majorization and trumping, two partial orders that have proven useful with respect to the entanglement transformation problem in quantum information,  with general Dirichlet polynomials, Mellin transforms, and completely monotone sequences.  We extend a basic majorization result to the more physically realistic infinite-dimensional setting through the use of generalized Dirichlet series and  Riemann-Stieltjes integrals.
\end{abstract}

\begin{keyword} majorization \sep entanglement-assisted local transformation \sep Dirichlet series \sep Mellin transformations\sep completely monotone functions

\MSC 81P40\sep  81P68\sep  	11F66

\end{keyword}

\end{frontmatter}

\section{Introduction}

The problem of entanglement transformation concerns the ability to transform from one pure state of some composite system to another, using only  local operations and classical communication (LOCC). Such  manipulations of entangled states have been characterized by way of the partial order of majorization in the finite-dimensional setting \cite{Nie99} and the infinite-dimensional setting \cite{OBNM08}.

Quantum mechanics is inherently infinite-dimensional by nature; although much work in quantum information theory is done under the restriction of finite dimensions, it is desirable to generalize to the infinite-dimensional setting (such generalizations are often  highly non-trivial).  In light  of an extension of Nielsen's result \cite{Nie99} to the infinite-dimensional setting by way of $\epsilon$-convertibility for LOCC \cite{OBNM08}, we extend the majorization  result of \cite{PP13} to the infinite-dimensional setting. Herein, we view the characterization of majorization put forward in \cite{PP13}  as a pure math inequality result; consequently, we do not consider the physical ramifications of infinite-dimensional majorization.

There are several definitions for infinite-dimensional majorization; we shall use that discussed in \cite{OBNM08}  as it best fits the physical descriptions of infinite-dimensional quantum states. That is, since the majorization condition of Nielsen involves vectors of Schmidt coefficients of pure states, the vectors are necessarily in $\ell_1$, and therefore there is no need to consider, for example, A.~Neumann's definition \cite{Neu99}, which allows for vectors in $\ell_\infty$.  
 Because we are working with positive trace-class operators, the  sum of the eigenvalues  that we are considering converges  to 1,   which leads to the promising realization that our Dirichlet series are well-behaved. 

\section{Dirichlet Series, Mellin Transforms, and Completely Monotone Functions}

\begin{definition}
Let $I$ be a real interval.  A function $f$ is said to be \emph{completely monotone} on $I$ if $(-1)^nf^{(n)}(x)\geq 0$ for all $x\in I$ and all $n=0, 1, 2, \dots$.
\end{definition}

 Bernstein's theorem on completely monotone functions states that a necessary and sufficient condition for a function $f$ to be completely monotone on $(0, \infty)$ is that $f$ is the Laplace transform of a positive measure $\mu$:
\[
f(s) = \int_0^\infty e^{-st} \,d\mu(t).
\]

We recall that the Mellin transform of a function $f$ on $(0,\infty)$ is the function $\phi(s)=\int_{0}^{\infty}f(t)t^{s-1}dt$.  The Mellin and Laplace transforms are closely related: if  $f\in L^1((0,\infty))$ is zero outside of $[0,1]$, then the Mellin transform of $f(x)$ is the Laplace transform of $f(e^{-x})$. It follows that if $f\in L^1((0,\infty))$ is zero outside of $[0,1]$, then the Mellin transform of $f$ is completely monotone on $(0,\infty)$ if and only if  $f$ is non-negative almost everywhere. We will use this fact in our proof of theorem \ref{thm:IDmaj}.

\begin{definition}
A \emph{generalized Dirichlet series} is a series of the form
\[
    \sum_{n=1}^{\infty}a_n e^{-\lambda_n s},
\]
where, herein, we take $a_n, s\in \R$ (in general, they can be complex) and $\{\lambda_n\}$ is a strictly increasing sequence of positive numbers that tends to infinity.

Letting $ \lambda_n=\log n$, we obtain a \emph{Dirichlet series}
\[
    \sum_{n=1}^{\infty}\frac{a_n}{n^s}.
\]
\end{definition}

\section{Majorization}

\begin{definition} \cite{OBNM08} $\ell_1^+$ consists of all $x=\{ x_n \}_{n=1}^{\infty}\in \ell_1(\mathbb{R})$ with the property that $x_n\geq 0$ for all $n\in \mathbb{N}$ and exactly one of the sets $\{ n\in \mathbb{N}:x_n>0 \}$ and $\{ n\in \mathbb{N}:x_n=0 \}$ is finite.\end{definition}

\begin{definition} \cite{OBNM08}
For any $x, y\in \ell_1^+$, we say that $x$ is majorized by $y$, written $x\prec y$, if
\begin{eqnarray*}
\sum_{i=1}^{k}x^{\downarrow}_{i}&\leq &\sum_{i=1}^{k}y^{\downarrow}_{i}\quad  k \in \N\\
\textnormal{and }\sum_{i=1}^{\infty}x^{\downarrow}_{i}&= &\sum_{i=1}^{\infty}y^{\downarrow}_{i},
\end{eqnarray*}
where $x^\downarrow$ symbolizes that the components of $x$ are rearranged in non-increasing order:  $x_i^\downarrow \geq x_{i+1}^\downarrow $ for all $i$.
\end{definition}

The results herein do not require vectors to be in  $\ell_1^+$ specifically; the extension to  $\ell_1$ does not pose a problem other than a slight modification of the definition of majorization (specifically, the addition of the requirement  $\sum_{i=1}^{k}x^{\uparrow}_{i}\geq \sum_{i=1}^{k}y^{\uparrow}_{i};  k \in \N$, where $x^\uparrow$ symbolizes that the components of $x$ are rearranged in non-decreasing order). However, due to possible convergence issues, our results likely cannot be extended to $\ell_\infty$.

Let $\m H$ be a Hilbert space and $\mathfrak{B}(\m H)$ be  the set of all (bounded) linear operators acting on $\m H$. An operator $\Phi:\mathfrak{B}(\m H)\rightarrow \mathfrak{B}(\m H')$ is \emph{trace-decreasing} if $\tr(\Phi(\rho))\leq \tr(\rho)$ for any $\rho\in \mathfrak{B}(\m H)$; an operator $\Phi$ is \emph{completely positive} if the induced mappings $\Phi_d:\mathcal{M}_d\otimes \mathfrak{B}(\m H)\rightarrow \mathcal{M}_d\otimes \mathfrak{B}(\m H')$ defined by $\Phi_d=\operatorname{id}_d\otimes
\Phi$ are positive for all $d\in \N$.

Consider two composite Hilbert spaces $\m H_A\otimes \m H_B$ and  $\m H_A'\otimes \m H_B'$.
A \emph{local operation} is mathematically described as a  trace-decreasing completely positive map:
\[\Phi_A \otimes \Phi_B:\mathfrak{B}(\m H_A\otimes \m H_B)\rightarrow \mathfrak{B}(\m H'_A\otimes \m H'_B)\] that acts separately on each component of the tensor product:
\[
\Phi_A \otimes \Phi_B = (\Phi_A \otimes \operatorname{id}_B)\circ(\operatorname{id}_A \otimes \Phi_B).
\]
Let $\operatorname{diag}(\mathcal{M}_d)$ denote the classical algebra of $d\times d$
diagonal matrices for some $d$.
\emph{Classical communication} is mathematically represented by
\begin{eqnarray*}
\Phi_A:& \mathfrak{B}(\m H_A) \rightarrow  \mathfrak{B}(\m H'_A) \otimes \operatorname{diag}(\mathcal{M}_d)\\
\textnormal{ and/or }\quad\Phi_B:& \mathfrak{B}(\m H_B) \otimes \operatorname{diag}(\mathcal{M}_d) \rightarrow  \mathfrak{B}(\m H'_B)
\end{eqnarray*}

Physically, LOCC is often described as a situation where two users, Alice and Bob, are spatially separated and can only perform local operations on their respective systems and can communicate classically by sending bits through a classical communication channel.

We adopt Dirac bra-ket notation: any unit vector in a Hilbert space $\m H$ will be written as a ``ket'' $\ket{\psi}$; its dual (complex-conjugate transpose) will be written as a ``bra'' $\bra{\psi}$.
\begin{definition}\cite{OBNM08}
Let $\ket{\psi}$ and  $\ket{\phi}$ be unit vectors (states) in $\m H_A\otimes \m H_B$. We say that $\ket{\psi}$ is  \emph{$\epsilon$-convertible} to $\ket{\phi}$ by LOCC if for any $\epsilon>0$, there exists an LOCC operation $\Lambda$ satisfying $||\Lambda(\ket{\psi}\bra{\psi})-\ket{\phi}\bra{\phi}||_{\tr}<\epsilon$, where $||\cdot||_{\tr}$ is the trace norm.
\end{definition}

The concept of $\epsilon$-convertibility allows for the extension of Nielsen's theorem \cite{Nie99}, which gives necessary and sufficient conditions for LOCC transformations,  to the infinite-dimensional setting:
\begin{theorem}\cite{OBNM08}
 $\ket{\psi}$ is  \emph{$\epsilon$-convertible} to $\ket{\phi}$ by LOCC if and only if $\lambda\prec \mu$, where $\mu$ refers to majorization in infinite-dimensional systems, and $\lambda$ and $\mu$ are the vectors of Schmidt coefficients of $\ket{\psi}$ and  $\ket{\phi}$, respectively.
\end{theorem}

An alternate characterization of infinite majorization found in \cite{OBNM08} will be particularly useful for us.  We first introduce the notation $z^+=\max (z,0)$ whenever $z\in \mathbb{R}$.

\begin{proposition} \cite[Appendix B]{OBNM08} \label{Obprop} Let $a=\{ a_n \}_{n=1}^{\infty}$ and  $b=\{ b_n \}_{n=1}^{\infty}$ both be in $\ell_1^+$.  Then $a\prec b$ if and only if the following two conditions are satisfied:

(1) $\sum_{j=1}^{\infty}(a_j-t)^+ \le \sum_{j=1}^{\infty}(b_j-t)^+$ for all $t>0$.

(2) $\sum_{i=1}^{\infty} a_i =\sum_{i=1}^{\infty} b_i$.

\end{proposition}

We are now in the position to introduce our result:

\begin{theorem}\label{thm:IDmaj} Let $a=\{ a_n \}_{n=1}^{\infty}$ and  $b=\{ b_n \}_{n=1}^{\infty}$ both be in $\ell_1^+$ and let $\zeta(s)=\sum_{n=1}^{\infty}b_n^s-\sum_{n=1}^{\infty}a_n^s$.
Then $a\prec b$ if and only if
\begin{enumerate}
\item[(i)] $\zeta(1)=0$;
\item[(ii)] $\frac{\zeta(s)}{s(s-1)}$ is completely monotone on $(1, \infty)$.
\end{enumerate}
\end{theorem}

\begin{proof}
For any $t>0$, we have: $\sum_{j=1}^{\infty}(a_j-t)^+ =\sum_{n:a_n>t} \int_t^{a_n} dx =\int_{t}^{\infty} A(x) dx$ where $A(x)$ is the number of elements of $\{a_n\}$ greater than or equal to $x$. (Note: since $a\in \ell_1$, $A(x)$ is finite on $(0,\infty )$ and $A(x)=0$ for $x$ large).

Let $Re$ $s>1$.  Now we can use the theory of Riemann-Stieltjes integration to represent our zeta function as follows:
\[\sum_{n=1}^{\infty}a_n^s=\int_{0}^{\infty}x^sdA=-s\int_{0}^{\infty}A(x)x^{s-1}dx,\]
 where the latter equality follows from integration by parts. We note that  $xA(x)\le \Vert a\Vert_1$ and hence $x^sA(x) \rightarrow 0$  when $x\rightarrow 0$.  We also have $x^sA(x)\rightarrow 0$ when $x\rightarrow \infty$ since $A(x)=0$ if $x\geq \sup_{n\in \mathbb{N}}$\ $a_n$.  This also means that these integrals are actually finite.

If we let $B(x)$ be the number of elements of $\{b_n\}$ greater than or equal to $x$, then we similarly obtain $\sum_{j=1}^{\infty}(b_j-t)^+ =\int_{t}^{\infty} B(x) dx$ and
\[\sum_{n=1}^{\infty}b_n^s=-s\int_{0}^{\infty}B(x)x^{s-1}dx.\]

Now suppose $\sum_{i=1}^{\infty} a_i =\sum_{i=1}^{\infty} b_i$, then using integration by parts we get:
\begin{eqnarray*}
\sum_{n=1}^{\infty}b_n^s-\sum_{n=1}^{\infty}a_n^s&=&-s\int_{0}^{\infty}(B(x)-A(x))x^{s-1}dx\\
&=&s(s-1)\int_{0}^{\infty}[\sum_{j=1}^{\infty}(b_j-t)^+-\sum_{j=1}^{\infty}(a_j-t)^+]t^{s-2}dt.
\end{eqnarray*}

Let $f(s)=\frac{\sum_{n=1}^{\infty}b_n^s-\sum_{n=1}^{\infty}a_n^s}{s(s-1)}=\frac{\zeta(s)}{s(s-1)}$.  We note that since $a,b\in \ell_1^+$, $f(s)$ converges absolutely on the half-plane $\{ z: Re$ $z>1\}$. Then if $\sum_{i=1}^{\infty} a_i =\sum_{i=1}^{\infty} b_i$, it follows that $f(s+1)$ is the Mellin  transform of $[\sum_{j=1}^{\infty}(b_j-t)^+-\sum_{j=1}^{\infty}(a_j-t)^+]$.  Therefore $f(s)$ is completely monotone on $(1,\infty)$ if and only if $f(s+1)$ is completely monotone on $(0,\infty)$ which occurs if and only if $[\sum_{j=1}^{\infty}(b_j-t)^+-\sum_{j=1}^{\infty}(a_j-t)^+]\geq 0$ for all positive $t$.  Since the latter statement plus the condition $\sum_{i=1}^{\infty} a_i =\sum_{i=1}^{\infty} b_i$ is equivalent to majorization by proposition \ref{Obprop}, $a\prec b$ if and only if both $\zeta(1)=0$ and $\frac{\zeta(s)}{s(s-1)}$ is completely monotone on $(1, \infty)$.
\end{proof}

This result can be seen as the infinite-dimensional version of \cite[theorem 3.1]{PP13}.  The proof is significantly different from the corresponding finite result, though it is similar in spirit.

\section{Discussion}

 Trumping \cite{JoPl99} is a partial order on vectors in $\R^d$ that generalizes the more familiar concept of majorization, allowing comparisons between a larger number of vectors than was possible under majorization. That is, trumping allows for a greater number of successful entanglement transformation  procedures. Moreover, these procedures are made possible by way of a catalyst state that remains unchanged after the LOCC protocol.

Partial results characterizing when trumping occurs include linking trumping with $\ell^p$ inequalities and  the von Neumann entropy \cite{Tur07} and  an equivalent characterization using a family of additive Schur-convex functions  \cite{Kli2004,Kli2007}, as well as an extension of these (equivalent) characterizations to higher order convex functions using generalized Dirichlet polynomials, Mellin transforms, and completely monotone functions \cite{PP13}. The latter greatly simplifies the proof of  the result found in \cite{Tur07}.

In infinite dimensions, we can define trumping as follows:
\begin{definition}
For any $x, y\in \ell_1^+$, we say that $x$ is  \emph{trumped} by $y$, written $x\prec_T y$,  if there exists a unit vector $c\in \ell_1^+$ with all positive components
 such that $x\otimes c\prec y\otimes c$.
\end{definition}
The \emph{catalyst}  $c$ is allowed to have infinite length in this setting, though it may be the case that it has finite length.

Motivated by the main result from \cite{AuNe08}, which characterized the closure of the set of all $x$ trumped by a fixed $y$, where $x$ and $y$ were finitely-supported infinite-dimensional vectors, we introduce the following conjecture:

\begin{conjecture}
 Let $\zeta(s)$ be a generalized Dirichlet series absolutely convergent on the half-plane $\{z: \operatorname{Re} z\ge 1\}$ with a simple zero at $s=1$.  Then $\zeta(s)$ is positive on $(1,\infty)$ if and only if there exists a generalized Dirichlet series $\zeta_2(s)\not\equiv 0$ with positive coefficients absolutely convergent on the half-plane $\{z:  \operatorname{Re} z\ge 1\}$ such that $\frac{\zeta(s)\zeta_2(s)}{s(s-1)}$ is completely monotone on $(1,\infty)$.
\end{conjecture}

If we let $a=\{ a_n \}_{n=1}^{\infty}$ and  $b=\{ b_n \}_{n=1}^{\infty}$ both be in $\ell_1^+$ and let $\zeta(s)=\sum_{n=1}^{\infty}b_n^s-\sum_{n=1}^{\infty}a_n^s$.
Then, using theorem \ref{thm:IDmaj} and the definition of infinite-dimensional trumping, the statement of this conjecture reduces to: $\zeta(s)$ is positive on $(1,\infty)$ if and only if there exists a catalyst $c=\{c_n\}_{n=1}^\infty\in \ell_1^+$ with corresponding generalized Dirichlet series $\zeta_2(s)=\sum_{n=1}^{\infty}c_n^s$  such that  $x\otimes c\prec y\otimes c$; that is, such that $x$ is trumped by $y$. The characteristic of $\zeta(s)$ being  positive on $\R$ with  simple zeros at $s=0$ and $s=1$    was shown to be equivalent to trumping  in the finite-dimensional setting \cite[proposition 3.3]{PP13}. This conjecture, if it holds, would be a Dirichlet series version of the Dirichlet polynomial result \cite[corollary 4.11]{PP13} and would effectively characterize infinite-dimensional trumping.

\section*{Acknowledgements}   R.P. was supported by NSERC Discovery Grant 400550.

\end{document}